\documentclass[12pt,reqno]{amsart}
\usepackage{amsmath,amssymb,amsthm}
\usepackage[dvipsnames]{xcolor}
\usepackage{comment}
\usepackage[shortlabels]{enumitem}
  \setlist[enumerate,1]{leftmargin=25pt}
  \setlist[itemize,1]{leftmargin=12pt}
  \setlist[description,1]{leftmargin=15pt}
\usepackage{float}\restylefloat{figure}

\usepackage[text={5.5in,10in},centering]{geometry}
\usepackage{graphicx}
\usepackage{hyperref}

\usepackage[leftmargin=16pt,vskip=6pt]{quoting}

\usepackage{setspace}
\usepackage{stix2}

\usepackage{url}

\newtheorem{theorem}{Theorem}[section]

\newtheorem{lemma}[theorem]{Lemma}
\newtheorem{proposition}[theorem]{Proposition}

\theoremstyle{definition}

\newtheorem{definition}[theorem]{Definition}

\newtheorem{proviso}[theorem]{Proviso}

\newcommand\At{\ensuremath{\raisebox{1pt}{\,\text{\tiny @}\,}}}

\newcommand\C{\ensuremath{\mathcal C}}
\newcommand\CELL{\ensuremath{\mathrm{CELL}}}
\newcommand\Co{\ensuremath{\mathbb C}}

\renewcommand\d{\ensuremath{\delta}}
\newcommand\dpa{\ensuremath{\d_{pa}}}

\newcommand\Entries{\ensuremath{\textrm{Entries}}}
\newcommand\Exits{\ensuremath{\textrm{Exits}}}
\newcommand\Gm{\ensuremath{\Gamma}}
\renewcommand\H{\ensuremath{\mathcal H}}

\newcommand\Inputs{\ensuremath{\textrm{Inputs}}}

\newcommand\Lpab{\ensuremath{L_{pab}}}
\newcommand\mo{{\raisebox{1pt}{\tiny$-$}}\text{\small1}}

\newcommand\ort{\ensuremath{\; \bot\;}}
\newcommand\Outputs{\ensuremath{\textrm{Outputs}}}
\newcommand\ox{\ensuremath{\otimes}}

\renewcommand\phi{\varphi}
\newcommand\po{{\raisebox{1pt}{\tiny$+$}}\text{\footnotesize1}}
\newcommand\qef{\hfill$\triangleleft$} 
\newcommand\qefhere{\tag*{$\triangleleft$}} 
\newcommand\Qmo{\ensuremath{Q_{\mo}}}
\newcommand\Qpo{\ensuremath{Q_{\po}}}

\newcommand\Rpab{\ensuremath{R_{pab}}}

\newcommand\U{\ensuremath{\mathcal U}}

\newcommand\x{\ensuremath{\times}}
\newcommand\Z{\ensuremath{\mathbb Z}}

\newcommand\If{\ensuremath{\mathrm{if }}}

\newcommand\braket[2]{\ensuremath{\langle#1\,|\,#2\rangle}}
\newcommand\ket[1]{\ensuremath{|#1\rangle}}
\newcommand\ketb[1]{\ensuremath{\big|#1\big\rangle}}

\newcommand\set[1]{\ensuremath{\{#1\}}}

\newcommand\qn{\ensuremath{q^{\uparrow}}}
\newcommand\Qn{\ensuremath{Q^{\uparrow}}}
\newcommand\qs{\ensuremath{q^{\downarrow}}}
\newcommand\Qs{\ensuremath{Q^{\downarrow}}}
\newcommand\vn{\ensuremath{v^{\uparrow}}}
\newcommand\vs{\ensuremath{v^{\downarrow}}}

\newcommand\Rn{\ensuremath{R^{\uparrow}}}
\newcommand\Ln{\ensuremath{L^{\uparrow}}}
\newcommand\Rs{\ensuremath{R^{\downarrow}}}
\newcommand\Ls{\ensuremath{L^{\downarrow}}}

\title[]
{Simple circuit simulations of\\ classical and quantum Turing machines}
\author{Yuri Gurevich}
\address{Computer Science and Engineering\\
University of Michigan\\
Ann Arbor, MI  48109-2121, U.S.A}
\email{gurevich@umich.edu}
\thanks{Partially supported by the US Army Research Office under W911NF-20-1-0297}
\author{Andreas Blass}
\address{Mathematics Department\\
University of Michigan\\
Ann Arbor, MI 48109--1043, U.S.A.}
\email{ablass@umich.edu}

\begin{document}

\begin{abstract}
We construct reversible Boolean circuits efficiently simulating reversible Turing machines. Both the circuits and the simulation proof are rather simple. Then we give a fairly straightforward generalization of the circuits and the simulation proof to the quantum case.
\end{abstract}

\maketitle
\thispagestyle{empty}

\section{Introduction}
\label{sec:intro}

The origin of this paper was our dissatisfaction with the standard definition of the uniform circuit complexity.
In the circuit computation model, a computational problem is associated with a family $C_1, C_2, \dots$ of circuits where each circuit $C_n$ handles $n$-bit inputs in the Boolean case or $n$-qubit inputs in the quantum case.
The most common uniformity condition requires that a single logspace\footnotemark\ Turing machine computes every $C_n$ given $n$ in unary notation.
\footnotetext{Another popular uniformity condition uses polynomial-time rather than logspace, but logspace uniformity is more common.}

This complexity-theoretic uniformity is important for theoretical purposes but is not used for practical programming.
Instead, circuits are usually described parametrically, and this suffices to guarantee uniformity.

A natural problem arises to formalize the notion of parametric uniformity so that one can deal with it theoretically.
In order to address this problem, we looked around for challenging examples.
In this connection, we came across a recent paper \cite{MW} by  Abel Molina and John Watrous describing circuits that simulate quantum Turing machines.
We tried to achieve such simulations by parametric circuits.
This paper reports our results in this direction as well as analogous results for reversible classical Turing machines.

We should emphasize though that the idea of parametric uniformity is just a motivation as far as this paper is concerned.
We intend to formalize parametric uniformity for Boolean and quantum circuits elsewhere \cite{G260}.
For the time being we use the notion of parametrized circuits informally.

It is well known that Turing machines are efficiently simulated by (uniform families of) Boolean circuits. In extended abstract \cite{Yao}, Andrew Yao sketched a proof that quantum Turing machines are efficiently simulated by (rather sophisticated) quantum circuits.
In paper \cite{MW}, Molina and Watrous give an elegant proof of Yao's result, but their circuits and their proofs are still rather involved.

A question arises whether simpler quantum circuits can simulate quantum Turing machines.
A related question is whether reversible Turing machines can be efficiently simulated by reversible Boolean circuits.

In this paper, we construct reversible Boolean circuits which efficiently simulate reversible Turing machines.
Both the circuits and the simulation proof are rather simple.
Then we give a fairly straightforward generalization of the circuits and the simulation proof to the quantum case.

\section{Preliminaries}\mbox{}
\label{sec:prelims}

The reader may want to skip this section and consult it only if and when needed.

Reversible Boolean circuits and quantum circuits are mathematical notions requiring precise definitions, and so are Boolean circuit algorithms and quantum circuit algorithms.
Not finding such definitions in popular textbooks, we produced definitions in \cite{G244} and \cite{G250}. To make this paper self-contained, we reproduce those definitions here with minor changes.
An added advantage is useful terminology that arises in the course of the definitions.

\begin{definition}
\label{def:syncirc}
A \emph{syntactic circuit} consists of the following components.
\begin{enumerate}
\item Disjoint finite sets of \emph{input nodes}, \emph{output nodes}, and \emph{gates}.
\item For each gate $G$, two disjoint nonempty finite sets, the set $\Entries(G)$ of the \emph{entries} of $G$ and the set $\Exits(G)$ of the \emph{exits} of $G$.\\
    The sets associated with any gate are disjoint from those associated with any other gate and from the sets in (1).
    The input nodes and gate exits will be called \emph{producers}. The gate entries and output nodes will be called \emph{consumers}.
\item A set of \emph{wires}.  Each wire begins at a producer and ends at a consumer.
    Each producer is the beginning of at least one wire.
    Each consumer is the end of exactly one wire.
\item A gate $G$ is a \emph{direct prerequisite} of a gate $G'$ if some wire goes from an exit of $G$ to an entry of $G'$.
    The direct prerequisite relation is acyclic.
    Its transitive closure is called the \emph{prerequisite} relation (without ``direct''). \qef
\end{enumerate}
\end{definition}

\begin{definition}
Consider a syntactic circuit.
\begin{itemize}
\item A \emph{gate bout} of the circuit is a nonempty set $B$ of gates which form an antichain (so that no gate $G\in B$ is a prerequisite of another gate $G'\in B$).
\item A \emph{schedule} over the circuit is a sequence
  $(B_1; B_2; B_3; \dots; B_T)$
  of disjoint gate bouts such that every gate $G$ belongs to some bout $B_t$ and all prerequisites of $G$ belong to earlier bouts $B_s$ where $s<t$.
\item A schedule $(B_1; B_2; B_3; \dots; B_T)$ is \emph{eager} if every $B_t$ comprises all gates whose prerequisites are all in $\bigcup_{s<t} B_s$ (so that in particular $B_1$ comprises all gates with no prerequisites). \qef
\end{itemize}
\end{definition}

\begin{definition}
\label{def:bcirc}
  A \emph{Boolean circuit} is a syntactic circuit together with an assignment, to each gate $G$, of a Boolean function
  \[ \beta_G: \set{0,1}^{\Entries(G)} \to \set{0,1}^{\Exits(G)}. \]
  An \emph{algorithm} over the Boolean circuit is given by a schedule over the underlying syntactic circuit.
  An algorithm is \emph{eager} if the schedule is eager. \qef
\end{definition}

Intuitively, the algorithm given by a schedule $(B_1; B_2; B_3; \dots; B_T)$ first fires all $B_1$ gates in parallel, then it fires all $B_2$ gates in parallel, and so on. Thus $T$ is the number of steps of the algorithm.
All algorithms over the same Boolean circuit \C\ compute the same Boolean function
  \[ \beta_{\C}: \set{0,1}^{\Inputs(\C)} \to \set{0,1}^{\Outputs(\C)} \]
but may differ in the number of steps.
Below, for simplicity, we restrict attention to eager algorithms.

\begin{proviso}
Below, by default, a Boolean circuit is identified with its eager algorithm.
\end{proviso}

\begin{definition}
\label{def:rcirc}
A Boolean circuit is \emph{reversible} if
\begin{enumerate}
\item each gate function $\beta_G$ is a bijection (so that in particular $G$ has equally many entries and exits), and
\item each provider has exactly one outgoing wire (so that the wires determine a bijection from producers to consumers). \qef
\end{enumerate}
\end{definition}

\section{Turing machines}
\label{sec:tm}

We presume a basic familiarity with Turing machines, in short TMs, and specify a particular TM species.
In this paper, by default, a Turing machine is deterministic and has one two-way infinite tape whose cells are indexed by integers.
We also assume that the tape head moves at every step.
A Turing machine is given by a finite set $Q$ of (control) states, a tape alphabet \Gm\ including a special blank symbol, and a partial transition function
\[
\delta : Q \x \Gamma \to Q \x \Gamma \x \set{-1,+1}. \qefhere
\]
If the machine is in state $p$, its head is scanning symbol $a$ in tape cell $i$, and $\delta(p,a) = (q, b, D)$, then in one step the machine will change state to $q$, write $b$ in the tape cell $i$, and move the tape head to cell $i+D$. (The change is vacuous if $q=p$ and the writing is vacuous if $b=a$.)
If $\delta(p,a)$ is undefined then the machine halts.
Initially, the tape is almost everywhere blank, i.e., the number of nonblank cells is finite.

A \emph{configuration} of a Turing machine $(Q,\Gm,\d)$ is a triple
\[ (p,i,X) \in Q\x\Z\x\Gm^\Z\]
where \Z\ is the set of integers and the function $X:\Z\to\Gm$ is almost everywhere blank.
The configuration $(p,i,X)$ means that the state is $p$, the tape head is in tape cell $i$, and $X$ is what is written on the tape.
(Of course, the blanks far to the left or the right
can be ignored.
But notationally it simpler to use all of \Z\ as the domain of $X$.)
Notice that the notion of configuration depends only on $Q$ and $\Gm$, not on \d.

Obviously, the transition function \d\ gives rise to a partial successor function on the configurations.
In detail, if $\d(p,X(i)) = (q,b,D)$, then the \emph{successor} of configuration $(p,i,X)$ is the configuration
\[ \Big(q,i+D,X(b\At i)\Big) \]
where the symbol \At\ indicates that symbol $b$ is written at cell $i$, so that $X(b\At i)(i)=b$ and $X(b\At i)(j) = X(j)$ for all $j\ne i$.
Accordingly $(p,i,X)$ is a \emph{predecessor} of $\Big(q,i+D,X(b\At i)\Big)$.
If $\d(p,X(i))$ is undefined, then $(p,i,X)$ has no successor.

\subsection{Reversible Turing machines}
\label{sub:rtm}

\begin{definition}\label{def:crev}
A Turing machine is \emph{reversible} if every configuration has at most one predecessor. \qef
\end{definition}

Arbitrary Turing machines are efficiently simulated by reversible ones; see \cite[\S5]{Morita}.
The criterion of reversibility of TM $(Q,\Gm,\d)$ in terms of \d\ is found in \cite[\S5.1.1.2]{Morita}. We present it in a form convenient for our purposes, and we use the following notation:
\begin{align*}
Q_{\mo} = \set{q: \d(p,a)=(q,b,-1)\ \text{ for some }\ p,a,b},\\
Q_{\po} = \set{q: \d(p,a)=(q,b,+1)\ \text{ for some }\ p,a,b}.
\end{align*}

\begin{proposition}\label{prp:crev}
A Turing machine $(Q,\Gm,\d)$ is reversible if and only if \d\ has the following two properties.
\begin{description}
\item[\tt Separability] $Q_{\mo}$ and $Q_{\po}$ are disjoint.
\item[\tt Injectivity] For every pair $(q,b)\in Q\x\Gamma$, there is at most one pair $(p,a)\in Q\x\Gamma$ such that $\d(p,a) = (q,b,D)$ for some $D$.
\end{description}
\end{proposition}

\begin{proof}
Let $X:\Z\to\Gamma$ be almost everywhere blank.

If the separability requirement fails, then $M$ is irreversible. Indeed, if
$\delta(p_1,a_1) = (q,b_1,\po)$ and $\delta(p_2,a_2) = (q,b_2,\mo)$, then configurations
\[ \Big(p_1,\mo,X(a_1\At\mo)(b_2\At\po)\Big) \ne
   \Big(p_2,\po,X(b_1\At\mo)(a_2\At\po)\Big), \]
have the same successor $(q,0,X(b_1\At\mo)(b_2\At\po))$.

If the separability requirement holds but the injectivity requirement fails, then $M$ is irreversible.
Indeed, if $\delta$ maps distinct pairs $(p_1,a_1), (p_2,a_2)$ to triples $(q,b,D_1), (q,b,D_2)$ respectively, then $D_1 = D_2$ by the separability requirement, and therefore configurations $(p_1,0,X(a_1\At0)) \ne (p_2,0,X(a_2\At0))$ have the same successor $(q,D_1,X(b\At0))$.

If both requirements hold, then $M$ is reversible. Indeed suppose that configurations $(p_1,i_1,X_1)$ and $(p_2,i_2,X_2)$ have the same successor $(q,j,X_3)$ where $q\in Q_D$.
Let $i=j-D$, so that $i_1 = i_2 = i$.
Then $X_1, X_2$ coincide with $X_3$ except possibly at $i$.
Then $\delta(p_1,X_1(i)) = \delta(p_2,X_2(i)) = (q,X_3(i),D)$. By injectivity, $p_1 = p_2$ and $X_1(i) = X_2(i)$, and so $(p_1,i_1,X_1) = (p_2,i_2,X_2)$.
\end{proof}

Taking into account the separability property, states in \Qmo\ and \Qpo\ will be called \emph{negative} and \emph{positive} respectively when the Turing machine in question is reversible.
If $Q_{\mo}\cup Q_{\po} \neq Q$, extend $Q_{\mo}, Q_{\po}$ arbitrarily so that they remain disjoint and $Q_{-1} \cup Q_{+1} = Q$.

\begin{lemma}\label{lem:total}
The transition function \d\ of a Turing machine $M = (Q,\Gm,\d)$ can be extended to a total function $\d^*$ in such a way that
\begin{enumerate}
\item if, starting with an initial configuration $(p,i,X)$, $M$ never halts, then the modified machine $M^* = (Q,\Gm,\d^*)$ works exactly as $M$ on that initial configuration, and
\item if $M$ is reversible then $M^*$ is reversible.
\end{enumerate}
\end{lemma}

\begin{proof}
Claim~(1) is obvious no matter how \d\ is extended to a total function $\d^*$.
To prove claim~(2), assume that $M$ is reversible. By the separability property and the convention just before the lemma, $Q_{\mo}$ and $Q_{\po}$ partition $Q$.

If $\d(p,a) = (q,b,D)$, let $\alpha(p,a) = (q,b)$.
By the injectivity property, $\alpha$ is an injective map from $Q\x\Gamma$ to $Q\x\Gamma$. Extend $\alpha$ to a bijection $\alpha^*$ from $Q\x\Gamma$ to $Q\x\Gamma$ in an arbitrary way.
If $\alpha^*(p,a) = (q,b)$ and $q\in Q_D$, define $\d^*(p,a) = (q,b,D)$.
Clearly $(Q,\Gm,\d^*)$ is reversible.
\end{proof}

\subsection{Looped Turing machine}

This paper is devoted to easy circuit simulation of reversible Turing machines and of quantum Turing machines.
Following the example of Molina and Watrous \cite{MW}, we separate concerns and work primarily with machines $(Q,\Gm,\d)$ where \d\ is total, so that the machine never halts.
We address the halting-related issues in \S\ref{sub:hang}.

Fix an arbitrary Turing machine $M_\infty = (Q,\Gm,\d)$ with total transition function \d.
$M$ never hangs.
Further, fix a positive integer $T$, and restrict attention to $T$-step computations of $M_\infty$ where the initial configuration is blank outside the segment $[-T,T]$ of the tape. During such a computation the head of $M_\infty$ is confined to the segment $[-T,T]$.

Truncate the tape of $M_\infty$ to the segment $[-T-1, T+1]$, bend that segment into a loop, and identify the end cells $-T-1$ and $T+1$, resulting in $N = 2T+2$ cells altogether.
Let $M = (Q,\Gm,\d)_T$, be the resulting looped-tape Turing machine.

View the cell indices as residues modulo $N$, and by default interpret index modifications $i\pm1$ modulo $N$.
If the state of $M$ is $p$, its head is scanning symbol $a$ in cell $i$ (modulo  $N$ of course), and $\delta(p,a) = (q, b, D)$,
then in one step $M$ changes its state to $q$, writes $b$ in cell $i$, and moves the head to cell $i+D$.
As far as our $T$-step computations are concerned, there is no difference between $M_\infty$ and $M$.

A \emph{configuration} of $M$ is a triple
\[ (p,i,X) \in Q \x (\Z/N) \x \Gm^{\Z/N}; \]
the role of \Z\ is played by the set $\Z/N$ of residues modulo $N$.
As in the case of $M_\infty$, \d\ induces a successor function on the configurations.
As in the case of $M_\infty$, $M$ is \emph{reversible} if every configuration has at most one predecessor.
Proposition~\ref{prp:crev} remains valid for $M$, with the same proof.
It follows that $M$ is reversible if and only if $M_\infty$ is reversible.

\section{Boolean circuit simulation}
\label{sec:bc}

We will simulate $T$-step computations of the looped-tape Turing machine $M = (Q,\Gm,\d)_T$ with a Boolean circuit of depth $2T$ where rows of gates are staggered like the bricks in the brick wall of Figure~\ref{fig:bwall}.

\begin{figure}[h]
\includegraphics[scale=.4,trim=40 480 40 54,clip]{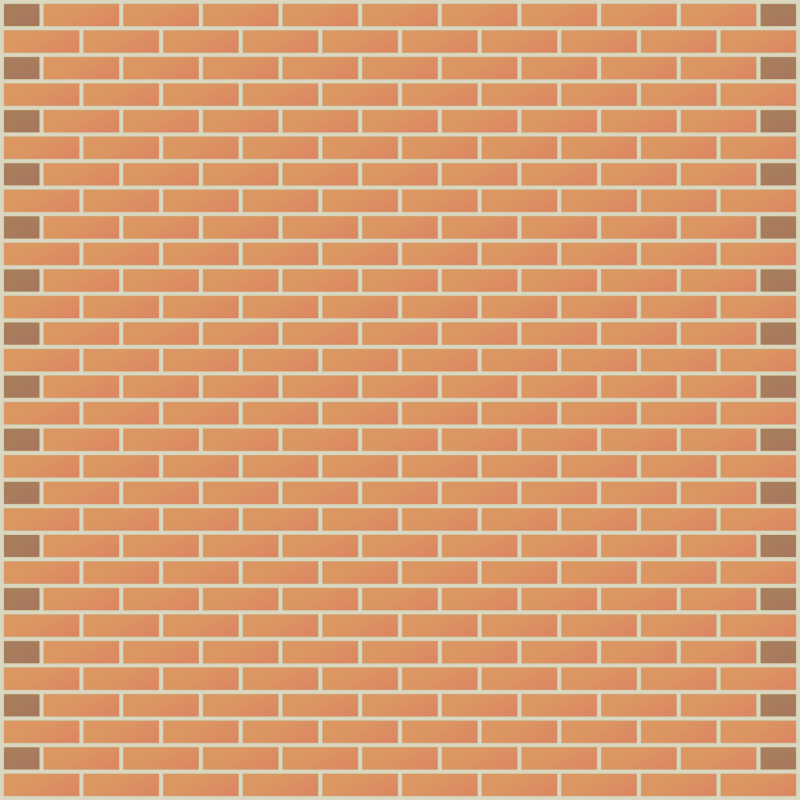}
\caption{A brick wall (American bond). Source: Wikimedia \cite{WikiBwall} }\label{fig:bwall}
\end{figure}

For our purposes, it is convenient to view the brick wall as being composed of half bricks, intended to represent tape cells of $M$ at various stages of computation.
Each half brick is uniquely determined by two coordinates.
\begin{description}
\item[\texttt{Horizontal coordinate}]
The brick wall comprises $N = 2T+2$ vertical columns of half bricks.
Number the vertical columns left to right using the residues $0,1,\dots,N-1$ modulo $N$.
This gives us the horizontal coordinate $h$.
\item[\texttt{Vertical coordinate}]
The wall has $2T$ rows of bricks.
Number them top down: $0,1, \dots, 2T-1$.
This gives us the vertical coordinate $v$.
\end{description}

A brick comprising the half bricks $(h,v)$ and $(h+1,v)$ represents a circuit gate denoted $G(h,v)$.
Notice that half brick $(h,v)$ is the left half of its brick if and only if $h+v$ is even.
Thus, if $h+v$ is even, then half brick $(h,v)$ belongs to $G(h,v)$; and, if $h+v$ is odd, then half brick $(h,v)$ belongs to gate $G(h-1,v)$.

The use of residues modulo $N$ for numbering the columns hints at a property of our brick wall which is not reflected in Figure~\ref{fig:bwall}. Our wall is bent into a cylinder, reflecting the loop structure of $M$.
The two half bricks at the ends of row 1 (or any other odd-numbered row) are the two halves of the same brick.

Think of the brick wall as a greedy circuit algorithm where each brick is a gate.
The algorithm is greedy in the sense that, at each step, it executes all the gates that can be executed.
All the gates of row 0 fire at the first step of the algorithm, all the gates of row 1 fire at the second step, and so on.

For every $t = 0, 1, \dots, T-1$, the subcircuit composed of rows $2t$ and $2t+1$ will transform any configuration
\[ (p,i,X) \in Q\x \Z/N \x \Gm^\Z \]
of the looped Turing machine $M$ into the successor configuration.
However, for technical reasons, the circuit works with richer data.
We introduce annotated versions $\qn, \qs$ of every state $q\in Q$.
It is presumed that the sets
\[ Q,\quad \Qn=\set{\qn: q\in Q},\quad
     \Qs=\set{\qs: q\in Q},\quad \set{0} \]
are disjoint. Let $Q^* = Q\sqcup \Qn\sqcup \Qs$.
Every row of the brick wall circuit will compute a unary operation on the set
\[ Q^*\x \Z/N \x \Gm^{\Z/N} \]
of \emph{extended configurations} of $M$.

\subsection{The brick function}
\label{sub:bf}\mbox{}

\noindent
All bricks (that is all brick gates) compute the same function $f$, the \emph{brick function}.
A \emph{cell datum} is an element of the set $(Q^*\sqcup\set{0}) \x \Gamma$.
Thus, a cell datum has the form $(q^*,a)$ or $(0,a)$ where $q^*\in Q^*$ and $a\in\Gm$.
In the first case, it is \emph{scanned}, and, in the second, it is \emph{unscanned}.
A \emph{brick datum} is a pair of cell data where at most one cell datum is scanned.
Thus, a brick datum has the form
$(q^*,a,0,b)$, $(0,a,q^*,b)$, or $(0,a,0,b)$, where $q^*\in Q^*$ and $a,b\in\Gm$.
In the first case, it is a \emph{one-head} datum, and, in the second, it is \emph{no-head} datum.

A brick (or brick gate) computes the brick function $f$ transforming any input brick datum to an output brick datum. Such inputs and outputs can, if desired, easily be coded as tuples of bits to make the circuits Boolean.

In the rest of this subsection, we define the brick function $f$. To begin, unsurprisingly
\begin{equation}\label{f0}
f(0,a,0,b) = (0,a,0,b).
\end{equation}
In other words, a no-head input is simply copied into the output.
(For brevity, we'll write $fx$ for $f(x)$.)
Suppose that the input brick datum $x$ includes the scanned cell. Three cases arise depending on whether the state symbol is in $Q$, \Qn, or \Qs.

First suppose that the state symbol is in $Q$, so that the input datum $x$ has the form $(p,a,0,c)$ or $(0,c,p,a)$. Let $\d(p,a) = (q,b,D)$. Then
\begin{equation}\label{f1}
\begin{split}
q\in Q_{+1} \implies f(p,a,0,c) = (0,b,\qn,c) \ \land\
                      f(0,c,p,a) = (0,c,\qs,b),\\
q\in Q_{-1} \implies f(p,a,0,c) = (\qs,b,0,c) \ \land\
                      f(0,c,p,a) = (\qn,c,0,b).
\end{split}
\end{equation}
Thus $\uparrow$ indicates successful execution of the \d-transition, and $\downarrow$ indicates incomplete execution when the brick updated the state and tape symbol but couldn't move the head.

Second, after a successful execution, all that remains is to remove $\uparrow$. We must, however, also define $f$ on 2-data with $\uparrow$ that cannot arise in the Turing machine simulation. In such ``garbage'' cases, we simply interchange the two cell data in the input brick datum.
\begin{equation}\label{f2}
\begin{split}
q\in Q_{+1} \implies f(\qn,c,0,b) = (q,c,0,b) \ \land\
                            f(0,b,\qn,c) = (\qn,c,0,b),\\
q\in Q_{-1} \implies f(0,b,\qn,c) = (0,b,q,c) \ \land\
                            f(\qn,c,0,b) = (0,b,\qn,c).
\end{split}
\end{equation}

Finally, after an incomplete execution, what remains is to move the head in the appropriate direction and remove $\downarrow$. Again, we define $f$ on ``garbage'' cases by interchanging the two cell data in the input brick datum.
\begin{equation}\label{f3}
\begin{split}
q\in Q_{+1} \implies f(\qs,b,0,c) = (0,b,q,c) \ \land\
                            f(0,c,\qs,b) = (\qs,b,0,c),\\
q\in Q_{-1} \implies f(0,c,\qs,b) = (q,c,0,b) \ \land\
                            f(\qs,b,0,c) = (0,c,\qs,b).
\end{split}
\end{equation}
This completes the definition of the brick function $f$.

\subsection{Reversibility}
\label{sub:rns}\mbox{}

\begin{theorem}
If Turing machine $M$ is reversible, then the brick wall circuit is reversible.
\end{theorem}

\begin{proof}
It suffices to prove that the brick gate is reversible, i.e., that the function $f$ defined above is a bijection.
Obviously, $f$ is bijective on no-head inputs.
So we may restrict attention to one-head inputs.
Since every injection from a finite set to a finite set of the same cardinality is a bijection, it suffices to prove that
$ fx = fy \implies x = y $
on one-head inputs.  To this end, assume $fx=fy$.

First we suppose that $fx, fy$ are defined by different equation systems \eqref{f1}--\eqref{f3}. By symmetry, we may assume that $fy$ is defined by a later equation system and that $fx$ is defined by one of the upper equations in \eqref{f1}, in \eqref{f2}, or in \eqref{f3}.
By inspection, there are only three cases:
\begin{enumerate}
\item $fx$ has the form $(0,b,\qn,c)$ and is defined the first upper equation in \eqref{f1}, while $fy$ is defined by the second lower equation in \eqref{f2}.
\item $fx$ has the form $(0,b,\qs,c)$ and is defined the second upper equation in \eqref{f1}, while $fy$ is defined by the second lower equation in \eqref{f3}.
\item $fx$ has the form $(q,c,0,b)$ and is defined by the first upper equation in \eqref{f2}, while $fy$ is defined by the first lower equation in \eqref{f3}.
\end{enumerate}
None of the three cases is possible because of the separability property of \d.

Thus $fx,fy$ are defined by the same equation system.
If $fx, fy$ are defined by system \eqref{f1}, then they are defined by the same equation in \eqref{f1}, and then we have $x=y$ by the injectivity property of \d. If $fx, fy$ are both defined by \eqref{f2} or are both defined by \eqref{f3}, then obviously $x=y$.
\end{proof}

\subsection{Simulation}

\begin{theorem}\label{thm:csim}
The brick wall circuit algorithm simulates any $T$-step computation of $M$.
\end{theorem}

\begin{proof}
It suffices to prove that, for every $t = 0, 1, \dots, T-1$, the two-row subcircuit, composed of rows $2t$ and $2t+1$,  transforms an arbitrary configuration $(p,i,X)$ of $M$ into the successor configuration. Since all these two-row circuits are isomorphic, we may restrict attention to the two-row circuit composed of rows 0 and 1.

Let $a=X(i)$.
Two cases arise depending on the parity of $i$. If $i$ is even, then the scanned cell $(p,a)$ is handled by brick $G(0,i)$ which works with input $(p,a,0,c)$ where $c=X(i+1)$; otherwise $(p,a)$ is handled by brick $G(0,i-1)$ which works with input $(0,d,p,a)$ where $d=X(i-1)$. By symmetry, we may assume that $i$ is even.

According to \eqref{f0}, no-head inputs are just copied into outputs. So we pay attention only to one-head inputs. Let $\d(p,a) = (q,b,D)$.

If $D=+1$, then brick $G(0,i)$ transforms $(p,a,0,c)$ to $(0,b,\qn,c)$, and thus row 0 transforms configuration $(p,i,X)$ to an extended configuration $(\qn,i+1,X(b\At i))$,
so that the scanned cell $(\qn,c)$ is handled by brick $G(1,i+1)$ which works with input $(\qn,c,0,X(i+2))$.
That brick transforms this input into $(q,c,0,X(i+2))$, and
therefore the resulting configuration is $(q,i+1,X(b\At i))$ as required.

If $D=-1$, then gate $G(0,i)$ transforms $(p,a,0,c)$ into $(\qs,b,0,c)$, and thus row 0 transforms configuration $(p,i,X)$ to an extended configuration $(\qs,i,X(b\At i))$, so that the scanned cell $(\qs,b)$ is handled by brick $G(1,i-1)$ which works with input $(0,d,\qs,b)$.
That brick transforms this input into $(q,d,0,b)$, and
therefore the resulting configuration is $(q,i-1,X(b\At i))$ as required.
\end{proof}

\subsection{Premature halting problem}\mbox{}
\label{sub:hang}

In the preceding part of this section, we have assumed that the looped Turing machine does not have halting computations. Every computation will proceed for (at least) the $T$ steps that we want to simulate with a Boolean circuit.
In the present subsection, we comment on the situation where the TM's computations could halt.

Of course, if $T$ steps are completed before the TM halts, then our
previous results apply to those $T$ steps, and subsequent halting
is irrelevant.  The real issue concerns premature halting, before
step $T$. We cannot simply cut off our circuit at the moment of TM's halting, as that moment may depend on the input.
Is circuit simulation of such a TM nevertheless possible?

First, let us agree on a reasonable definition of Boolean circuit simulation in this situation.
A simulation of a TM by a circuit \C\ includes an encoding of initial TM states as inputs to \C\ and a decoding of outputs of \C\ as information about TM configurations.
If the TM arrives at a halting configuration $C = (p,i,X)$, then the circuit's output should code $C$ or, more generally, a required amount of information about $C$, e.g.\ a version of $C$ where $X$ is replaced by its restriction to a vicinity of the cell $i$.
Whatever it is, the required amount of information as well as a coding scheme should be agreed on ahead of time.

In the absence of any reversibility considerations, this situation is easy to address.
Define an extension $M' = (Q',\Gm',\d')_T$ of our Turing machine $M = (Q,\Gm,\d)_T$ so that any halting $M$ configuration $C$ gets a successor $M'$ configuration $C'$ which preserves the required information about $C$.
For example, $C'$ could be $C$ itself.
(This requires allowing the TM's head to stay at the same cell, but our work above can easily be extended to cover this option.)
Instead of halting, $M'$ is idling, pretending to compute without doing much.
Then let the circuit simulate $M'$.

But making the extended machine $M'$ reversible is not obvious.
Let $\Sigma$ be the set of halting $M$ configurations and all their $M'$ descendants, i.e., successors, successors of successors, etc.
Then $\Sigma$ is closed under the successor operation.
Assume that $M'$ is reversible, so that the successor operation is injective.
Since the total number of $M'$ configurations is finite, the successor operation is a bijection on $\Sigma$.
Accordingly, the successor of a nonhalting configuration cannot be halting.
Thus, every nonhalting initial configuration yields a nonhalting computation.

What can one do? Well, keep $\Sigma$ from being closed under the successor operation. Starting from a halting $M$ configuration $C$, the extended machine $M'$ will eventually halt, but not before marching around for sufficiently many steps.

Let\ $M' = (Q',\Gm',\d')_T$ where $Q'$ and $\Gm'$ are disjoint unions $Q\,\sqcup\, \set{p': p\in Q}$ and $\Gm\, \sqcup\, \set{a':a\in\Gm}$ respectively, and let $\d'$ be a function from $Q'\x\Gm$ to $Q' \x \Gm' \x \set{-1,+1}$ such that if $p\in Q$ and $a\in\Gm$ then
\begin{align*}
\d'(p,a) &=
\begin{cases}
  \d(p,a)    &\text{if }\ \d(p,a)\ \text{ is defined},\\
  (p',a',+1) &\text{if }\ \d(p,a)\
              \text{ is undefined and }p\in Q_{+1},\\
  (p',a',-1) &\text{if }\ \d(p,a)\
              \text{ is undefined and }p\in Q_{-1},\\
\end{cases} \\
&\text{and}\\
\d'(p',a) &=
\begin{cases}
  (p',a',+1) &\text{if }p\in Q_{+1},\\
  (p',a',-1) &\text{if }p\in Q_{-1},\\
\end{cases}
\end{align*}
Recall that, just before Lemma~\ref{lem:total}, we arranged that $Q_{-1}$ and $Q_{+1}$ partition $Q$.

$M'$ is reversible because \d\ has the separability and injectivity properties. It remains to check that, on no input, $M'$ halts before step $T$. Fix some input $I$. If $M$ doesn't halt before step $T$ on input $I$, we are done. So suppose that, after $t<T$ steps, $M$ arrives at a halting configuration $(p,i,X)$ on input $I$.

By symmetry, we may assume without loss of generality that $p\in Q_{+1}$.
Abbreviate $X(j)$ to $a_j$; we have that $\d(p,a_i)$ is undefined.
The $M'$ successor of $(p,i,X)$ is $(p',i+1,X(a_i'\At i))$, whose $M'$ successor is $(p', i+2, X(a_i'\At i, a_{i+1}'\At(i+1))$, and so on. Eventually, after $N$ steps (counting from the halting of $M$), $M'$ has converted every $a_j$ into $a_j'$ and halts.

It follows that, on input $I$, $M'$ will perform at least $T$ steps without halting. Its configuration after $T$ steps is $(p',i + (T-t),X')$ where
\[X'(j) =
\begin{cases}
  a_j' &\If\ i\le j < i+(T-t),\\
  a_j  &\text{otherwise}.
\end{cases} \]
The halting configuration $C=(p,i,X)$ of $M$ is easily recoverable from the configuration  $C'=(p',i + (T-t),X')$; just move the head to the cell of the first primed symbol and then remove all the primes.

Extend the transition function $\d'$ of $M'$ to a total function
\[ \d^*: Q'\times \Gm' \to  Q'\times \Gm'\times \set{-1,+1}\]
as in the proof of Lemma~\ref{lem:total}. The resulting TM $M^* = (Q',\Gm',\d^*)$ is reversible and so, by Theorem~\ref{thm:csim}, an appropriate brick wall circuit \C\ simulates $T$-step computations of $M^*$. But $M^*$ works exactly as $M'$ on the initial configurations of $M'$, and so \C\ simulates $T$-step computations of $M'$.

\section{Quantum Turing machine}
\label{sec:qc}

Quantum Turing machines, in short QTMs, were introduced by David
Deutsch \cite{Deutsch}, thoroughly studied by Ethan Bernstein and Umesh Vazirani \cite{BV}, and recently presented elegantly by Abel Molina and John Watrous in article \cite{MW} which is our main reference on QTMs.

Like a classical Turing machine, a quantum Turing machine is also given by a triple $(Q,\Gm,\d)$ where $Q$ and \Gm\ are disjoint finite sets, but now \d\ is a function
\[ \d: Q\x\Gm \to \Co^{Q\x\Gm\x\set{-1,+1}}. \]
Quantum Turing machines are also subject to a unitarity requirement to be introduced below.
We abbreviate $\d(p,a)(q,b,D)$ to $\dpa(q,b,D)$.

A \emph{classical configuration} of $M$ is a triple
\[ (p,i,X) \in Q \x \Z \x \Gm^\Z. \]
This notion coincides with the notion of configuration in \S\ref{sec:tm} for the same $Q$ and $\Gm$; the change in \d\ makes no difference.
Initially a QTM is in a classical configuration, just like a classical TM with the same states and alphabet (but of course later the QTM can enter a quantum superposition of classical configurations).

Our goal is a simple quantum-circuit simulation of $T$-step computations of a given QTM
\[ M = (Q,\Gm,\d).\]
Here $T$ is an arbitrary but fixed positive integer.
It is presumed that, initially, the head is in cell 0, and the tape is blank outside the segment $[-T,T]$.
As in the classical case, we may assume without loss of generality that $M$ has a looped tape whose cells are numbered by residues modulo $N = 2T+2$. Accordingly, a \emph{classical configuration} of the looped-tape QTM $M$ is a triple
\[ (p,i,X) \in Q \x (\Z/N) \x \Gm^{\Z/N}. \]

To simplify the exposition, we separate concerns and abstract from two issues, as in \cite{MW}.
\begin{enumerate}
\item The transition function \d\ of our looped-tape QTM is total.
\item We place no computability restrictions on the complex numbers appearing in \d. ``Any computationally offensive properties possessed by \d\ will be inherited by the quantum circuits that result from the simulation,'' \cite[\S3]{MW}.
\end{enumerate}

\subsection{The Hilbert space of the QTM}
\mbox{}

For any finite set $S$, let $\lBrack S\rBrack$ be the Hilbert space with an orthonormal basis \set{\ket{s}: s\in S}, the \emph{canonical basis} for $\lBrack S\rBrack$.
If $B_1, B_2$ are canonical bases for Hilbert spaces $\H_1, \H_2$ respectively then the \emph{canonical basis} for $\H_1 \ox \H_2$ is the orthonormal basis composed of vectors $\ket{b_1b_2} = \ket{b_1}\ox \ket{b_2}$ where \ket{b_1}, \ket{b_2} range over $B_1, B_2$ respectively.

Let CELL be the Hilbert space $\big\lBrack Q\sqcup\set{0}\big\rBrack \ox \lBrack\Gm\rBrack$, and consider
the tensor product $\CELL^{\ox\, \Z/N}$.
Each classical configuration $(p,i,X)$ of the looped-tape QTM $M$ is represented by a $\CELL^{\ox\, \Z/N}$ vector
\begin{equation}\label{canonvec}
\begin{aligned}
\ket{p,i,X} =\ &\ket{r_0a_0} \ox \ket{r_1a_1} \ox \cdots\ox \ket{r_{N-1}a_{N-1}}\ \text{ where }\\
&r_i=p,\ \text{ all other $r_j=0$,
and all } a_j = X(j).
\end{aligned}
\end{equation}
Intuitively, these vectors are the ones where exactly one cell contains a head.
The Hilbert space \H\ of $M$ is the subspace of
$\CELL^{\ox\, \Z/N}$ generated by vectors \ket{p,i,X}
which form the \emph{canonical basis} for $\H$.
Thus, as in \cite{MW}, the canonical orthonormal basis for the Hilbert space of the QTM of interest consists of vectors \ket{p,i,X} representing classical configurations of the QTM.

Finally, we can complete the definition of quantum Turing machine  $M = (Q,\Gm,\d)$. It is required that the transformation
\[\U_M\ket{p,i,X} = \sum_{q,b,D}\d_{pX(i)}(q,b,D)\,\ket{q,i+D,X(b\At i)}\]
be unitary. This $U_M$ represents any one step in the computation of $M$.

\subsection{Transition function}\mbox{}

View every $\d_{pa}$ as a vector
\[ \dpa = \sum_{q,b,D} \d_{pa}(q,b,D) \ket{q,b,D} \]
in the Hilbert space $\lBrack Q\rBrack\ox \lBrack\Gm\rBrack\ox \lBrack\set{-1,+1}\rBrack$. Let
\[ \ketb{\Lpab} = \sum_q \dpa(q,b,-1)\ket{q,b,{-1}}, \quad
\ketb{\Rpab} = \sum_q \d_{pa}(q,b,+1)\ket{q,b,{+1}}, \]
so that $\d_{pa} = \sum_b\ketb{\Rpab} + \sum_b\ketb{\Rpab}$.

\begin{proposition}
\d\ has the following three properties.
\begin{align*}
&\text{Unit length}:
  &&|\d_{pa}| = 1
  &&\text{for all } p,a.\\
&\text{Orthogonality}:
  &&\d_{p_1a_1}\ort \d_{p_2a_2}
  &&\text{for all } (p_1,a_1)\ne(p_2,a_2).\\
&\text{Separability}:
  &&\ketb{L_{p_1a_1b_1}} \ort \ketb{R_{p_2a_2b_2}}
  &&\text{for all } p_1,a_1,b_1,p_2,a_2,b_2.
\end{align*}
\end{proposition}
\noindent
The three claims%
\footnote{Orthogonality and separability here are the analogs of injectivity and separability in Proposition~\ref{prp:crev}.}
form a part of Theorem~5.3 in \cite{BV}, but it seems useful to spell out a detailed proof.

\begin{proof}
Let $X$ be an arbitrary element of $\Gm^{\Z/N}$; for example, $X$ may be everywhere blank.

\texttt{Unit length.} If $C$ is a classical configuration of the form $(p,0,X(a\At0))$, then $\U_M\ket C = \sum_{q,b,D}\d_{pa}(q,b,D)\ket{q,D,X(b\At0)}$.
Since $\U_M$ is unitary, we have
\[
 1 = |\U_M\ket C|^2 = \sum_{q,b,D} \big(\d_{pa}(q,b,D)\big)^*\d_{pa}(q,b,D)  = |\d_{pa}|^2.
\]

\texttt{Orthogonality}.
If $(p_1,a_1)\ne(p_2,a_2)$, then canonical basis vectors \ket{p_1,0,X(a_1\At0)} and \ket{p_2,0,X(a_2\At0)} of $\H$ are orthogonal, and so are their $\U_M$ images:
\[\sum_{q,b,D}\d_{p_1a_1}(q,b,D)\ket{q,D,X(b\At0)} \ort
\sum_{q,b,D}\d_{p_2a_2}(q,b,D)\ket{q,D,X(b\At0)}.
\]
Hence
\[
 0 = \sum_{q,b,D} \big(\d_{p_1a_1}(q,b,D)\big)^* \d_{p_2a_2}(q,b,D) =
 \braket{\d_{p_1a_1}}{\d_{p_2a_2}}.
\]

\texttt{Separability.}
For any $p_1,a_1,b_1,p_2,a_2,b_2$, the canonical basis configurations \ket{p_1,1,X(a_1\At1,b_2\At{-1})} and \ket{p_2,-1,X(b_1\At1,a_2\At{-1})} are orthogonal, and therefore so are their $\U_M$ images:
\begin{align*}
&\sum_{q,b,D}\d(p_1,a_1)(q,b,D)\ket{q,1+D,X(b\At1,b_2\At{-1})} \ort\\
&\sum_{q,b,D}\d(p_2,a_2)(q,b,D)\ket{q,-1+D,X(b_1\At1,b\At{-}1)}.
\end{align*}
A summand in the upper sum are orthogonal to every summand in the lower sum unless $b=b_1$ and $D=-1$. Similarly, a summand in the lower sum is orthogonal to every summand in the upper sum unless $b=b_2$ and $D+1$.
Setting $Y = X(b_1\At1,b_2\At{-}1)$, we have
\[ \sum_q\d(p_1,a_1)(q,b_1,-1)\ket{q,0,Y} \ort
   \sum_q\d(p_2,a_2)(q,b_2,+1)\ket{q,0,Y}.
\]
The separability property follows.
\end{proof}

\section{Quantum circuit simulation}
\label{sec:qc}

We convert the reversible Boolean brick wall circuit, used to simulate $T$-step computations of a reversible Turing machine, to a quantum circuit simulating $T$-step computations of the quantum Turing machine $M$ of the previous section.
As in \S\ref{sec:bc}, we use annotated versions of states of $M$. Again, sets
\[ Q,\quad \Qn=\set{\qn: q\in Q},\quad
     \Qs=\set{\qs: q\in Q},\quad \set{0} \]
are disjoint and $Q^* = Q\sqcup \Qn\sqcup \Qs$.
The obvious bijections $q\mapsto \qn$ and $q\mapsto \qs$ give rise to obvious isomorphisms from $\lBrack Q \rBrack$ to $\lBrack \Qn \rBrack$ and to $\lBrack \Qs \rBrack$ respectively.

The quantum brick circuit works with an extension $\H^*$ of the Hilbert space $\H$ of our QTM $M$.
Let $\CELL^*$ be the Hilbert space $\big\lBrack Q^*\sqcup\set{0}\big\rBrack \ox \lBrack \Gm\rBrack$.
The extended space $\H^*$ is the subspace of the tensor product
$ \left(\CELL^*\right)^{\ox\, \Z/N} $
generated by vectors
\begin{equation}\label{hmstar}
\begin{aligned}
\ket{q^*,i,X} =\ &\ket{r_0a_0} \ox \ket{r_1a_1} \ox \cdots\ox \ket{r_{N-1}a_{N-1}}\quad\text{where}\\
& q^*\in Q^*,\ i\in\Z/N,\ X:\Z/N\to\Gm,\\
& r_i = q^*,\ \text{ all other $r_j=0$, and all }\ a_j=X(j).
\end{aligned}
\end{equation}
Every row of the quantum brick wall circuit will compute a linear operator on $\H^*$.

To simplify the exposition, our quantum circuit works  with qudits in Hilbert spaces $\lBrack Q^*\rBrack$ and $\lBrack\Gm\rBrack$. Obviously such quantum circuits can be easily simulated by qubit-based circuits.

\subsection{The brick operator}\mbox{}

Let BRICK be the subspace of tensor product $\CELL^* \ox \CELL^*$ generated by vectors
\[
\ket{q^*,a,0,b},\ \ket{0,a,q^*,b},\ \ket{0,a,0,b}\quad
\text{where }\ q^*\in Q^*\ \text{and } a,b\in\Gm.
\]
Vectors \ket{q^*,a,0,b} and \ket{0,a,q^*,b} are \emph{one-head vectors}, and vectors \ket{0,a,0,b} are \emph{no-head vectors}.

All gates of our quantum circuit are copies of a single quantum gate, called the \emph{(quantum) brick gate}, which computes a linear operator $U$ on BRICK.
By linearity, it suffices to define $U$ on a set of vectors spanning BRICK. We do that in equations \eqref{u0}--\eqref{u3}.

We start with no-head inputs.
$U$ just copies any such input into its output.
\begin{equation}\label{u0}
 U\ket{0,a,0,b} = \ket{0,a,0,b}
\end{equation}

Let $p$ range over $Q$.
The operator $U$ updates the state and the tape symbol according to \d, but it may or may not be able to move the head.
Processing any input \ket{0,c,p,a}, a gate $G(h,k)$ succeeds with the left but not the right moves; processing \ket{p,a,0,c}, it succeeds with the right but not the left moves.
In either case, successes are marked with $\uparrow$, and failures are marked with $\downarrow$:
\begin{equation*}
\begin{aligned}
U\ket{0,c,p,a}
&= \sum_{qb} \dpa(q,b,\mo)\,\ketb{\qn,c,0,b}
 + \sum_{qb} \dpa(q,b,\po)\,\ketb{0,c,\qs,b}, \\
U\ket{p,a,0,c}
&= \sum_{qb} \dpa(q,b,\po)\,\ketb{0,b,\qn,c}
 + \sum_{qb} \dpa(q,b,\mo)\,\ketb{\qs,b,0,c}.
\end{aligned}
\end{equation*}

Recall that $\ketb{L_{pab}} = \sum_q \dpa(q,b,\mo)\ket{q}$ and
$\ketb{R_{pab}} = \sum_q \dpa(q,b,\po)\ket{q}$.
Let \ketb{\Ln_{pab}}, \ketb{\Rn_{pab}} and \ketb{\Ls_{pab}}, \ketb{\Rs_{pab}} be the images of \ketb{L_{pab}}, \ketb{R_{pab}} under the obvious isomorphisms from $\lBrack Q \rBrack$ to $\lBrack \Qn \rBrack$ and to $\lBrack \Qs \rBrack$ respectively.
Then the two equations can be rewritten thus:
\begin{equation}\label{u1}
\begin{aligned}
U\ket{0,c,p,a}
&= \sum_b \ketb{\Ln_{pab},c,0,b}
 + \sum_b \ketb{0,c,\Rs_{pab},b},\\
U\ket{p,a,0,c}
&= \sum_b \ketb{0,b,\Rn_{pab},c}
 + \sum_b \ketb{\Ls_{pab},b,0,c}.
\end{aligned}
\end{equation}

Let $L,R$ be the subspaces of $\lBrack Q\rBrack$ generated by  vectors \ketb{L_{pab}} and by vectors \ketb{R_{pab}}.
Let subspaces \Ln, \Rn of $\lBrack \Qn\rBrack$ and subspaces \Ls, \Rs\ of $\lBrack \Qs\rBrack$ similarly be the images of $L, R$ under the obvious isomorphisms.
By the separability property of \d, $L$ and $R$ are orthogonal.
It follows that $\Ln \ort \Rn$ and $\Ls \ort \Rs$.
Let \ket{v} range over the unit vectors of $\lBrack Q\rBrack$.


The gates $G(h-1,k+1)$ and $G(h+1,k+1)$ complete $G(h,k)$'s successful execution by removing the success marks. To achieve that, we define

\begin{equation}\label{u2}
\begin{aligned}
U\ket{0,d,\vn,c} &= \ket{0,d,v,c}   &&\text{if }\ket v\in L,\\
U\ket{\vn,c,0,d} &= \ket{v,c,0,d}   &&\text{if }\ket v\in R.
\end{aligned}
\end{equation}

We must also define $U\ket{0,d,\vn,c}$ when $\vn\ort L$ and define $U\ket{\vn,c,0,d}$ when $\vn\ort R$. As in \S\ref{sub:bf}, we swap cell data.

\begin{equation}\label{u2'}
\begin{aligned}
U\ket{0,d,\vn,c} &= \ket{\vn,c,0,d} &&\text{if }\ket v \ort L,\\
U\ket{\vn,c,0,d} &= \ket{0,c,\vn,d} &&\text{if }\ket v \ort R.
\end{aligned} \tag{\ref{u2}$'$}
\end{equation}

The gates $G(h-1,k+1)$ and $G(h+1,k+1)$ perform the moves that gate $G(h,k)$ couldn't accomplish or swap cell data.
\begin{equation}\label{u3}
\begin{aligned}
U\ket{0,d,\vs,c} &= \ket{v,c,0,d}   &&\text{if }\ket v\in L,\\
U\ket{\vs,c,0,d} &= \ket{0,d,v,c}   &&\text{if }\ket v\in R.
\end{aligned}
\end{equation}

\begin{equation}\label{u3'}
\begin{aligned}
U\ket{0,d,\vs,c} &= \ket{\vs,c,0,d} &&\text{if }\ket v \ort L,\\
U\ket{\vs,c,0,d} &= \ket{0,c,\vs,d} &&\text{if }\ket v \ort R.
\end{aligned} \tag{\ref{u3}$'$}
\end{equation}

\subsection{Unitarity}

\begin{proposition}
There is a unique linear operator $U$ on BRICK, satisfying equations \eqref{u0}--(\ref{u3}$'$), and this $U$ is unitary.
\end{proposition}

\begin{proof}
Let $B$ be the set of vectors which appear as arguments of $U$ in equations \eqref{u0}--(\ref{u3}$'$).
The uniqueness of $U$ is obvious because $B$ spans BRICK.
Inspection of the equations shows the existence of $U$ because the equations don't clash when extended linearly.
(In more detail, choose orthonormal bases for $L$, for $R$, and for the orthogonal complement of $L+R$ in $\lBrack Q\rBrack$. Restrict \ket{v} in \eqref{u1}--(\ref{u3}$'$) to range over these bases.
Then $U$ is defined on an orthonormal basis for BRICK, and the existence of a linear extension is clear.
And of course \eqref{u1}--(\ref{u3}$'$) continue to hold as written.)
It remains to prove the unitarity of $U$.

Since $B$ spans BRICK, it suffices to prove that $U$ preserves the inner product on $B$, i.e., that $\braket{Ux}{Uy} = \braket xy$ holds for all $x,y\in B$.
This is obvious if $x$ or $y$ is a no-head vector; in that case $\braket{Ux}{Ux} = \braket xx =1$ and $\braket{Ux}{Uy} = \braket xy = 0$ for every $y\ne x$ in $B$.
In the rest of the proof, we assume that both, $x$ and $y$, are one-head vectors. By symmetry, we may assume that $Uy$ is defined by the same equation as $Ux$ or by a later equation.

Case~1: $Ux$ is defined by \eqref{u1}.
By the symmetry, we may assume that $Ux$ is defined by the upper equation.

If $Uy$ is defined by one of the following equations, then trivially \braket xy = 0. In all these cases, trivially \braket{Ux}{Uy} = 0 except for the following two cases.
One is the equation in (\ref{u2}$'$) where $\ket{v} \ort L$; we have $\ketb{\vn} \ort \Ln$ and therefore \braket{Ux}{Uy} = 0.
The other is the equation in (\ref{u3}$'$) where $\ket{v} \ort R$; we have $\ketb{\vs} \ort \Rs$, and therefore \braket{Ux}{Uy} = 0.

So suppose that $Uy$ is also defined by the upper equation in \eqref{u1}.
If $x=y$ then, by the unit-length property of \d, we have \braket{Ux} {Ux} = 1 = \braket xx.
Otherwise, by the orthogonality property of \d, \braket{Ux}{Uy} = 0 = \braket xy.

Case~2: $Ux$ is defined by \eqref{u2}.
By the symmetry, we may assume that $Ux$ is defined by the upper equation, where $\ket{v}\in L$.

If $Uy$ is defined by one of the following equations, then trivially \braket xy = 0.
In all these cases, trivially \braket{Ux}{Uy} = 0 except for the case of the lower equation in \eqref{u3}, where $\ket{v}\in R$. Use the fact that $L\ort R$.

So suppose that $Uy$ is also defined by the upper equation in \eqref{u2}.
We have \braket{Ux}{Uy} = \braket xy by the obvious isomorphism from $\lBrack Q \rBrack$ to $\lBrack \Qn \rBrack$.

Case~2$'$: $Ux$ is defined in (\ref{u2}$'$).
By the symmetry, we may assume that $Ux$ is defined by the upper equation.
If $Uy$ is defined by one of the following equations, then trivially \braket xy = 0 = \braket {Ux}{Uy}.
If $Uy$ is also defined by the upper equation in (\ref{u2}$'$), then obviously \braket xy = 0 = \braket {Ux}{Uy}.

Case~3: $Ux$ is defined in \eqref{u3}.
By the symmetry, we may assume that $Ux$ is defined by the upper equation.
If $Uy$ is defined by one of the following equations, then trivially \braket{Ux}{Uy} = \braket xy.
So suppose that $Uy$ is also defined by the upper equation in \eqref{u3}.
We have \braket{Ux}{Uy} = \braket xy by the obvious isomorphism from $\lBrack Q \rBrack$ to $\lBrack \Qs \rBrack$.

Case~3$'$ is similar to case~2$'$.
\end{proof}

A quantum circuit may give rise to a number of quantum algorithms. Semantically, all these algorithms are equivalent \cite{G250}. In particular, if one of them computes a unitary transformation then they all compute the same unitary transformation. But the number of computation steps may differ from one algorithm to another. The greedy algorithm, at each step, executes all the gates that can be executed. Below by circuit computations we mean the computations of the greedy algorithm of the circuit in question.

\subsection{Simulation}

\begin{theorem}
The brick wall circuit simulates any $T$-step computation of our quantum Turing machine $M$.
\end{theorem}

\begin{proof}
It suffices to prove that, for every $t = 0, 1, \dots, T-1$, the two-row subcircuit of our quantum brick circuit, composed of rows $2t$ and $2t+1$,  computes the unitary transformation $\U_M$ on Hilbert space $\H$.
Since all these two-row circuits are isomorphic, it suffices to prove that the two-row circuit $C$ composed of rows 0 and 1 computes $\U_M$.
By linearity, it suffices to prove that $C$ transforms any canonical basis vector \ket{p,i,X} of $\H_M$ into
\[
\U_M\ket{p,i,X} = \sum_{q,b,D}\dpa(q,b,D)\ket{q,i+D,X(b\At i)}
\]
where $a = X(i)$.

Let $c=X(i+1)$ and $d=X(i-1)$.
Two cases arise depending on the parity of $i$.
If $i$ is even, then the crucial cell $(p,a)$ is handled by
gate $G(0,i)$ which works with input $(p,a,0,c)$, and, if $i$ is odd, then $(p,a)$ is handled by gate $G(0,i-1)$ which works with input $(0,d,p,a)$. By symmetry, we may assume that $i$ is even.

Recall that $\ket{\Lpab} = \sum_q\dpa(q,b,\mo)\ket{q}$ and $\ket{\Rpab} = \sum_q\dpa(q,b,\po)\ket{q}$.
Let $\lambda = \set{b: \braket{\Lpab}{\Lpab}>0}$ and $\rho = \set{b: \braket{\Rpab}{\Rpab}>0}$.
Then $\U_M\ket{p,i,X}$ is the sum
\[  \sum_{b\in\lambda}\ketb{\Lpab,i-1,X(b\At i)} +
    \sum_{b\in\rho}\ketb{\Rpab,i+1,X(b\At i)}\]
of vectors \ketb{\Lpab,i-1,X(b\At i)} and \ketb{\Rpab,i+1,X(b\At i)}.

According to \eqref{u0}, given any canonical basis vector, only one gate of the upper row does real work. The other gates just copy their inputs to their outputs.
By the lower equation in \eqref{u1}, the first row of gates transforms \ket{p,i,X} into
\[ \sum_{b\in\lambda} \ket{\Ls_{pab},i,X(b\At i)} +
   \sum_{b\in\rho}     \ket{\Rn_{pab},i+1,X(b\At i)}.
\]
which, according to \eqref{u0},\eqref{u2} and \eqref{u3}, is transformed by the second row of gates into
\[ \sum_{b\in\lambda} \ket{\Lpab,i-1,X(b\At i)} +
   \sum_{b\in\rho}     \ket{\Rpab,i+1,X(b\At i)}
\]
as required.
\end{proof}


\newpage
\begin{thebibliography}{99}

\bibitem{BV}
Ethan Bernstein and Umesh Vazirani,
``Quantum complexity theory,''
\emph{SIAM Journal of Computing} 26:5, 1411--1473, 1997,
\url{https://epubs.siam.org/doi/10.1137/S0097539796300921}

\bibitem{Deutsch}
David Deutsch, ``Quantum theory, the Church–Turing principle and the universal quantum computer,''
\emph{Proceedings of the Royal Society A} 400 (1985), 97–-117,
\url{https://doi.org/10.1098/rspa.1985.0070}

\bibitem{G244}
Andreas Blass and Yuri Gurevich,
``Circuits: An abstract viewpoint,''
EATCS Bulletin 131 June 2020,
\url{https://arxiv.org/abs/2006.09488}

\bibitem{G250}
Andreas Blass and Yuri Gurevich,
``Quantum circuits with classical channels and the principle of deferred measurements,"
\url{https://arxiv.org/abs/2107.08324}

\bibitem{G260}
Yuri Gurevich and Andreas Blass,
``Parametrized circuits'' (tentative title),
in preparation

\bibitem{MW}
Abel Molina and John Watrous,
``Revisiting the simulation of quantum Turing
machines by quantum circuits,''
\emph{Proceedings of the Royal Society A} 475:20180767 (2019),
\url{http://dx.doi.org/10.1098/rspa.2018.0767}

\bibitem{Morita}
Kenichi Morita,
``Theory of reversible computing,''
Springer Japan KK 2017

\bibitem{WikiBwall}
Wikimedia Commons,
``Brickwork in stretching bond,''
\url{https://commons.wikimedia.org/w/index.php?&oldid=479882946},
accessed 5 August 2021

\bibitem{Yao}
Andrew Chi-Chih Yao Yao,
``Quantum circuit complexity,``
in Proceedings of 34th Annual IEEE Symposium on Foundations
of Computer Science, IEEE Press, Piscataway, NJ, 1993, 352-–361.

\end{thebibliography}
\end{document}